\documentclass{cccg23} 
\usepackage[T1]{fontenc}
\usepackage[utf8]{inputenc}

\usepackage{amsfonts,amsmath,amssymb, complexity}
\usepackage{cite,enumitem,microtype,graphicx}
\newtheorem{observation}[theorem]{Observation}
\newtheorem{corollary}[theorem]{Corollary}

\usepackage{todonotes}

\theoremstyle{definition}
\newtheorem{definition}[theorem]{Definition}

\title{Geometric Graphs with Unbounded Flip-Width\thanks{Research initiated at the 10th Workshop on Geometry and Graphs, Feb. 3–10, 2023, Bellairs Research Institute, Barbados. We thank the workshop participants (especially Jit Bose) and Szymon Toruńczyk for helpful conversations on this work.}}
\author{David Eppstein\thanks{Department of Computer Science, University of California, Irvine. Research supported in part by NSF grant CCF-2212129.} \and Rose McCarty\thanks{Department of Mathematics, Princeton University. Research supported by the NSF under grant DMS-2202961.}}

\date{ }

\begin{document}
\maketitle  

\begin{abstract}
We consider the flip-width of geometric graphs, a notion of graph width recently introduced by Toruńczyk. We prove that many different types of geometric graphs have unbounded flip-width. These include interval graphs, permutation graphs, circle graphs, intersection graphs of axis-aligned line segments or axis-aligned unit squares, unit distance graphs, unit disk graphs, visibility graphs of simple polygons, $\beta$-skeletons, 4-polytopes, rectangle of influence graphs, and 3d Delaunay triangulations.
\end{abstract}

\section{Introduction}

\emph{Flip-width} is a new and very general notion of width in graphs, defined by Szymon Toruńczyk~\cite{Tor-23} using a cops-and-robbers game on graphs, and intended to capture graph structure in a way that allows for efficient parameterized algorithms. It is hoped that testing whether a given graph models a first-order formula in the logic of graphs can be solved efficiently when parameterized by flip-width and formula size, although currently this is known only for more limited classes of graphs~\cite{DreMahSie-23}.

Beyond potential algorithms, another purpose of flip-width is to unify incompatible notions of graph width, including bounded expansion and of twin-width. A graph family has \textit{bounded expansion} if all shallow minors of its graphs are sparse~\cite{NesOss-12}. It has bounded \textit{twin-width} if its graphs can be reduced to one vertex by contracting pairs of vertices so that the subgraph of pairs of contracted vertices with inconsistent adjacencies maintains bounded degree throughout the contraction process~\cite{BonKimTho-JACM-22}. The sparse graph families of bounded flip-width are exactly the families of bounded expansion, and every graph family of bounded twin-width has bounded flip-width~\cite{Tor-23}. It is easy to construct graph families that have bounded flip-width but neither bounded expansion nor bounded twin-width, such as the family of the graphs that are either subcubic or cographs. The subcubic graphs have bounded expansion but unbounded twin-width~\cite{BonGenKimThoWat} while cographs reverse these inclusions. This union is not very natural, though; subcubic graphs and cographs have little in common. Can we find a natural family of graphs with bounded flip-width, but neither bounded twin-width nor bounded expansion?

Natural candidates include geometric graphs, whose vertices come from points or other simple objects in a geometric space, and whose edges are defined by simple geometric relations between these objects. However, planar graphs, and bounded-ply disk intersection graphs in bounded dimensions have bounded expansion~\cite{MilTenThuVav, DvoNor16}, as do sparse intersection graphs of connected subsets of a surface~\cite{Lee17, DvoNor16}. To find the examples we seek, we need non-sparse graphs. Geometric graph theory contains many examples of highly structured but non-sparse graph families. Do any have bounded flip-width?

In this work we provide a negative answer for many standard geometric graphs. We find a class of subgraphs common to these graphs, which we call ``interchanges'' and which provide a winning strategy for a robber in the cops-and-robbers game used to define flip-width. A graph family that includes arbitrarily large interchanges has unbounded flip-width. Using this idea we show that interval graphs, permutation graphs, circle graphs, intersection graphs of simply-intersecting axis-aligned line segments, intersection graphs of axis-aligned unit squares, unit distance graphs, unit disk graphs, visibility graphs of simple polygons, $\beta$-skeletons, rectangle of influence graphs, and the graphs of 4-polytopes all have unbounded flip-width. We provide a different construction showing that the graphs of 3-dimensional Delaunay triangulations have unbounded flip-width. 

For many of these graphs we prove more strongly that the radius-1 flip-width is unbounded and that these graphs are monadically independent, a related concept in the logic of graphs. A similar approach was used by Hliněný, Pokrývka, and Roy~\cite{HliPokRoy19} to prove hardness of first-order model checking on graph classes with a specific type of interchange, which they call the ``consecutive neighbourhood representation property''. Other hardness results, as well as some efficient algorithms, have been obtained for various geometric graphs in~\cite{BonChaKim-IPEC-22, GHKOST13, GHLOORS15, HliPokRoy19}. The proofs of such hardness results typically imply that the flip-width is unbounded, using the following key fact; a transduction of a class of bounded flip-width also has bounded flip-width~\cite{Tor-23}. Beyond extending these results to more graph classes, our approach has the advantages of only using first concepts, and of providing a concrete robber strategy and a specific bound on the radius.


\section{Cops and robbers}

Like treewidth~\cite{SeyTho-JCTB-93} and bounded expansion~\cite{Tor-23}, flip-width can be defined using a certain cops-and-robbers game. The games for treewidth and expansion involve ``cops with helicopters'', chasing a robber on a graph. The cops can occupy a limited number of graph vertices (initially, none); the robber can choose any starting vertex. In each time step, the cops announce where they will move next, the robber moves to escape them on a path through currently-unoccupied vertices, and then the cops fly directly to their new locations. The cops win by landing on the robber's current vertex, and the robber wins by evading the cops indefinitely. The treewidth of a graph is the maximum number of cops that a robber can evade, moving arbitrarily far on each move~\cite{SeyTho-JCTB-93}. A family of graphs has bounded expansion if and only if, for some function $f$, a robber who moves $\le r$ steps per move can be caught by $f(r)$ cops~\cite{Tor-23}.

The same game can be described differently. Instead of occupying a vertex, the cops set up roadblocks on all edges incident to it. On each move, the cops announce which vertices will be blockaded next. Then, the robber moves along un-blockaded edges. Finally, the cops remove their current blockades and put up new blockades at the announced locations. The cops win by leaving the robber at an isolated vertex, unable to move. Flip-width is defined in the same way, but with more powerful cops. Instead of blockading a single vertex, they may ``flip'' any subset of vertices. This complements the subgraph induced by that subset: pairs of adjacent vertices become non-adjacent, and vice versa. Blockading a single vertex, for instance, takes two flips: one flip of the vertex and its neighbors, and one of just the neighbors. The first flip disconnects the given vertex, and the second restores its neighbors' adjacencies. It doesn't matter in which order these two flips (or any set of flips) is performed.

In the flipping game used to define flip-width, at any move, the cops may perform a limited number of flips, initially none. The robber chooses an arbitrary starting vertex. In each move, the cops announce their next set of flips. The robber moves on a path in the current flipped graph, to evade these flips. Then, the cops undo their current flips and perform the flips that they announced. The cops win by leaving the robber at an isolated vertex, unable to move, and the robber wins by avoiding this fate indefinitely. A family of graphs has \textit{bounded flip-width} if, for some function $f$, $f(r)$ flips per move suffice to catch a robber who moves $\le r$ steps per move. Similarly, the \textit{radius-$r$ flip-width} of a graph is the least number of flips required to catch a robber who moves $\le r$ steps. Bounded flip-width implies bounded radius-$r$ flip-width, but not vice versa; for instance, subdivisions of complete graphs have bounded radius-1 flip-width but unbounded flip-width. Conversely, unbounded radius-$r$ flip-width implies unbounded flip-width, but not vice versa.

Because flipping can simulate blockading, graphs of bounded treewidth also have bounded flip-width. However, graphs of unbounded treewidth may have bounded flip-width. For instance, all planar graphs have bounded flip-width but the planar graphs have unbounded treewidth.

\section{Escaping through interchanges}

In the treewidth game, escape strategies for the robber are modeled graph-theoretically by havens, certain functions from subsets of vertices to connected components of the subgraph formed by their removal~\cite{SeyTho-JCTB-93}. In the same spirit, and using terminology following a road network metaphor, we define \emph{interchanges}, structures in a graph which can be used to define an escape strategy for the robber in the flipping game.

\begin{figure}[t]
\centering\includegraphics[width=0.45\columnwidth]{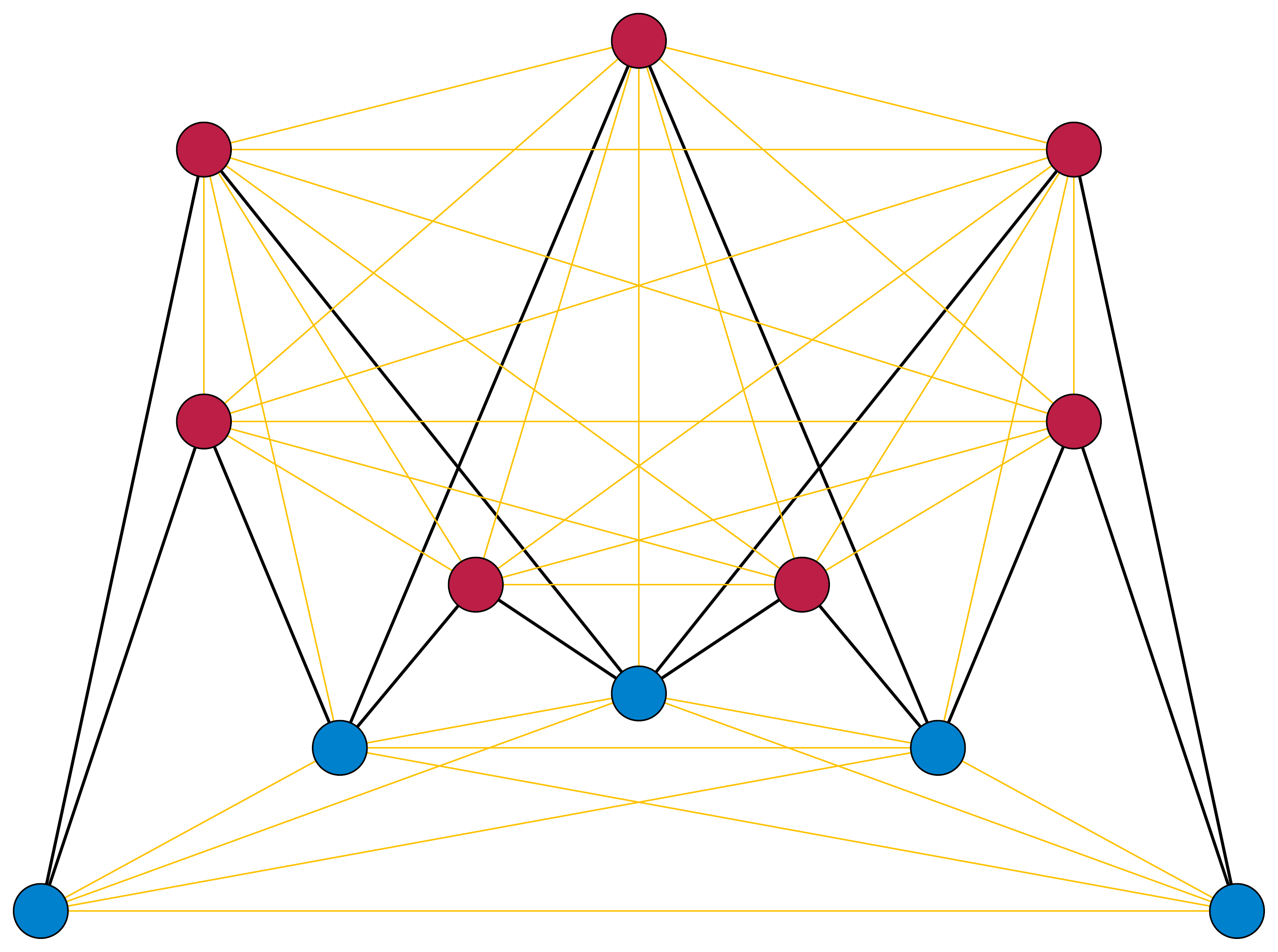}
\caption{An interchange of order five, with lanes in blue and ramps in red. The yellow edges are optional.}
\label{fig:5-interchange}
\end{figure}

\begin{definition}
An interchange of order $n$ consists of:
\begin{itemize}\setlength\itemsep{0em}
\item A linear sequence of $n$ designated vertices, which we call \emph{lanes}.
\item More designated vertices, called \emph{ramps}. Each ramp is associated with two lanes, and each two lanes that are $\le n-3$ steps apart in the sequence have a ramp. (We do not require ramps for farther-apart lanes because they would not be of use to the robber.)
\item An edge between each ramp and its two lanes.
\item Optional edges between any two lanes or between any two ramps. These will be unused by the robber. Making them optional, rather than specifying their presence or absence, allows us to construct geometric realizations without worrying about whether the construction includes these edges.
\item For a ramp that connects lanes $x$ and $y$, optional edges to other lanes between $x$ and $y$ in the sequence. Edges to lanes outside that range are not allowed.
\end{itemize}
\end{definition}

\noindent\cref{fig:5-interchange} depicts an example.

\begin{definition}
Let $\mathcal{F}$ be a collection of flips that could be made in the flipping game (a family of sets of vertices of a given graph). We define two lanes of an interchange to be \textit{equivalent under $\mathcal{F}$} if, for every flip $F$ in $\mathcal{F}$, either both lanes belong to $F$ or both are omitted from $F$.
\end{definition}

\begin{lemma}
\label{lem:triple-to-triple}
Let $a,b,c$ and $d,e,f$ be two disjoint triples of lanes such that, for a collection of flips $\mathcal{F}$, all lanes in $\{a,b,c\}$ are equivalent under $\mathcal{F}$, and all lanes in $\{d,e,f\}$ are equivalent under $\mathcal{F}$. Then, after the flips in $\mathcal{F}$ are made, the flipped interchange contains at least one two-edge lane--ramp--lane path between $\{a,b,c\}$ and $\{d,e,f\}$.
\end{lemma}

\begin{proof}
Assume (by swapping the triples if necessary) that the lane $b$ is before the lane $e$. Then these six lanes contain the four-lane subsequence $a,b,e,f$. If ramp $be$ is flipped with respect to the equivalent lanes $\{a,b,c\}$, it becomes adjacent to $a$; otherwise it remains adjacent to $b$. If ramp $be$ is flipped with respect to the equivalent lanes $\{d,e,f\}$ it becomes adjacent to $f$; otherwise it remains adjacent to $e$. In all cases this ramp connects at least one lane in $\{a,b,c\}$ to at least one lane in $\{d,e,f\}$.
\end{proof}

\begin{lemma}
\label{lem:triple-to-many}
Suppose that distinct lanes $a$, $b$, and $c$, in an interchange of order $n$, are equivalent under a collection of flips $\mathcal{F}$.
Then, in the flipped interchange, at least one of $a$, $b$, or $c$ has paths of length two to at least $\frac{1}{3}\bigl(n - 2^{|\mathcal{F}|+1}-3\bigr)$-many other lanes.
\end{lemma}

\begin{proof}
Under the flips in $\mathcal{F}$, there are $2^{|\mathcal{F}|}$ equivalence classes of lanes. By \cref{lem:triple-to-triple}, each equivalence class has at most two lanes (other than $a$, $b$, and $c$) that are not connected by a two-edge path to at least one of $a$, $b$, and $c$, because three disconnected but equivalent lanes would contradict the lemma. The total number of these disconnected vertices is at most $2^{|\mathcal{F}|+1}$; the remaining $n - 2^{|\mathcal{F}|+1}-3$ vertices have two-edge paths to at least one of $a$, $b$, or $c$. Even if each were connected to exactly one of $a$, $b$, or $c$, and even if these connections were evenly distributed between $a$, $b$, and $c$, the statement of the lemma would hold. Multiple connections or uneven distribution of connections only increases the largest of the three numbers of connections among $a$, $b$, and $c$.
\end{proof}

\begin{theorem}
\label{thm:radius-2}
Suppose that cops and a robber play the radius-$2$ flipping game with $t$ flips per move on a graph that includes an interchange of order $n=2^{t+3}+3$. Then the robber can win by moving at each step (including the initial step) to a lane that maximizes the number of lanes that will be reachable after the announced flips.
\end{theorem}

\begin{proof}
In an interchange of this size, by \cref{lem:triple-to-many}, every three equivalent lanes include one connected to
$\ge 2^{t+1}$ other lanes by two-edge paths; it can reach $\ge 2^{t+1}+1$ lanes including itself by paths of length $\le 2$. Among every $2^{t+1}+1$ lanes, at least three are equivalent. By induction, at each step, the robber has a choice of $\ge 2^{t+1}+1$ lanes to move to, among which three are equivalent, and therefore can move to a lane that will continue to reach at least $2^{t+1}+1$ lanes after the announced flips.
\end{proof}

\begin{figure}[t]
\includegraphics[width=\columnwidth]{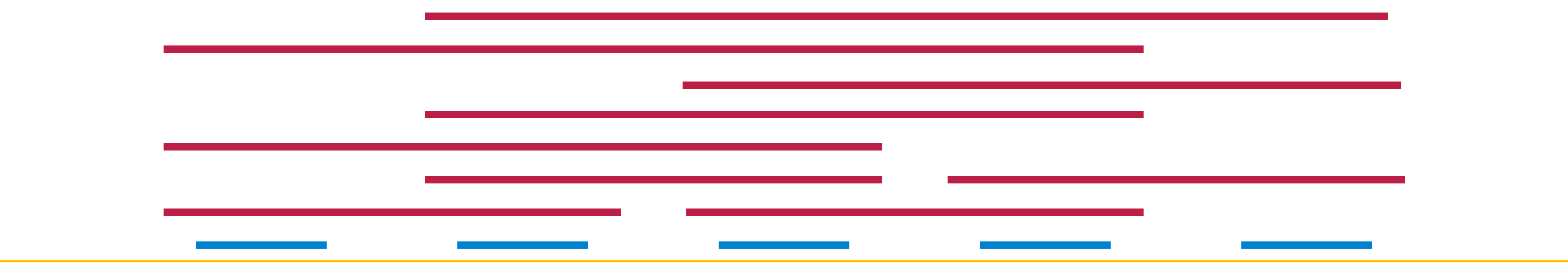}
\caption{Representing an interchange using the intervals of an interval graph or interval containment graph.}
\label{fig:interval-interchange}
\end{figure}

\begin{corollary}
\label{cor:interchange}
A class of graphs that contains arbitrarily large interchanges does not have bounded flip-width.
\end{corollary}

In Appendix~\ref{sec:monadicNIP} we show, more strongly, that graph classes with large interchanges are not monadically dependent, a property that generalizes both classes of bounded flip-width~\cite{Tor-23} and classes that are nowhere-dense~\cite{AdlAdl14}.

\section{Geometric graphs}

The geometric graphs known to have bounded flip-width include the unit interval graphs (which more strongly have bounded twin-width~\cite{BonGenKim-ICALP-01}) and the intersection graphs of disks of bounded ply in any fixed dimension (which more generally have bounded expansion~\cite{MilTenThuVav, DvoNor16}). We prove that many other classes of geometric graphs do not have bounded flip-width, by finding large interchanges in them and applying \cref{cor:interchange}.

\begin{theorem}
\label{thm:interval}
The interval graphs, permutation graphs, circle graphs, and intersection graphs of axis-aligned line segments (no two collinear) have unbounded flip-width.
\end{theorem}

\begin{proof}
We construct intervals representing an arbitrarily large interchange, with short disjoint intervals for each lane and long intervals spanning multiple lanes for each ramp (\cref{fig:interval-interchange}). The intersection graph of these intervals is an interval graph forming the interchange, with all optional lane--ramp edges included. The interval containment graph, having a vertex per interval and an edge whenever one interval contains another, differs only in some optional ramp--ramp edges. Interval containment graphs are the same as permutation graphs, and are a subclass of circle graphs~\cite{BraLeSpi-99}. For axis-aligned line segments, lift the ramp intervals to distinct $y$-coordinates, and replace the lane intervals by tall vertical segments.
\end{proof}

\begin{figure}[t]
\centering\includegraphics[width=0.7\columnwidth]{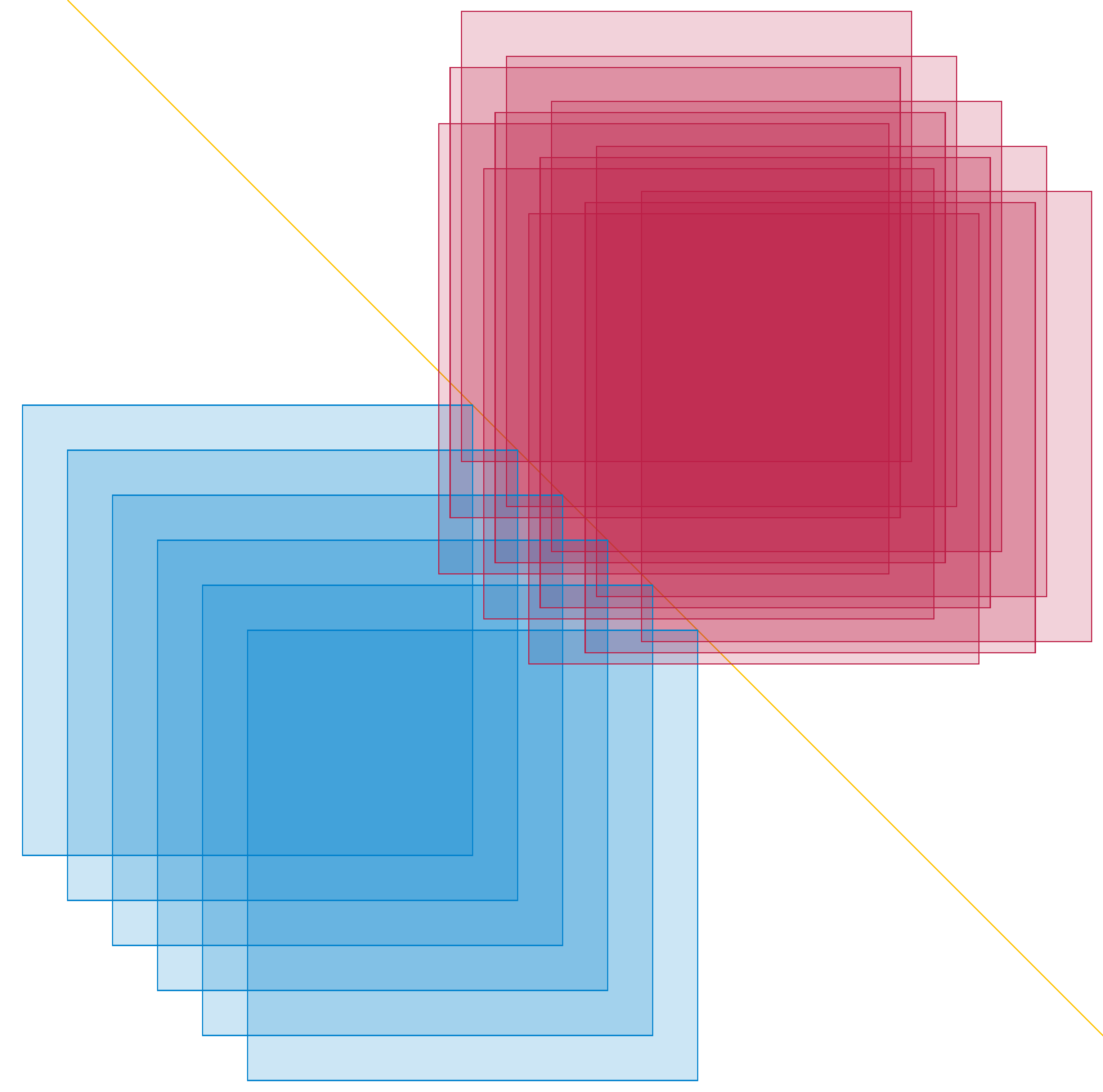}
\caption{Representing an interchange using axis-aligned unit squares.}
\label{fig:square-interchange}
\end{figure}

The graph classes in \cref{thm:interval} are monadically independent~\cite{BonChaKim-IPEC-22}, from which unbounded flip-width follows, but without a bound on the robber's escape radius.

\begin{theorem}
The intersection graphs of axis-aligned unit squares have unbounded flip-width.
\end{theorem}

\begin{proof}
Place the intervals of \cref{thm:interval} on a diagonal line, scaled to have length less than $\sqrt2$. Represent lanes by squares below this line, intersecting the line in the given interval, and represent ramps by squares above this line (\cref{fig:square-interchange}). The resulting unit square intersection graph may have additional lane--lane and ramp--ramp intersections, but these only create optional edges.
\end{proof}

\begin{figure}[t]
\centering\includegraphics[width=0.9\columnwidth]{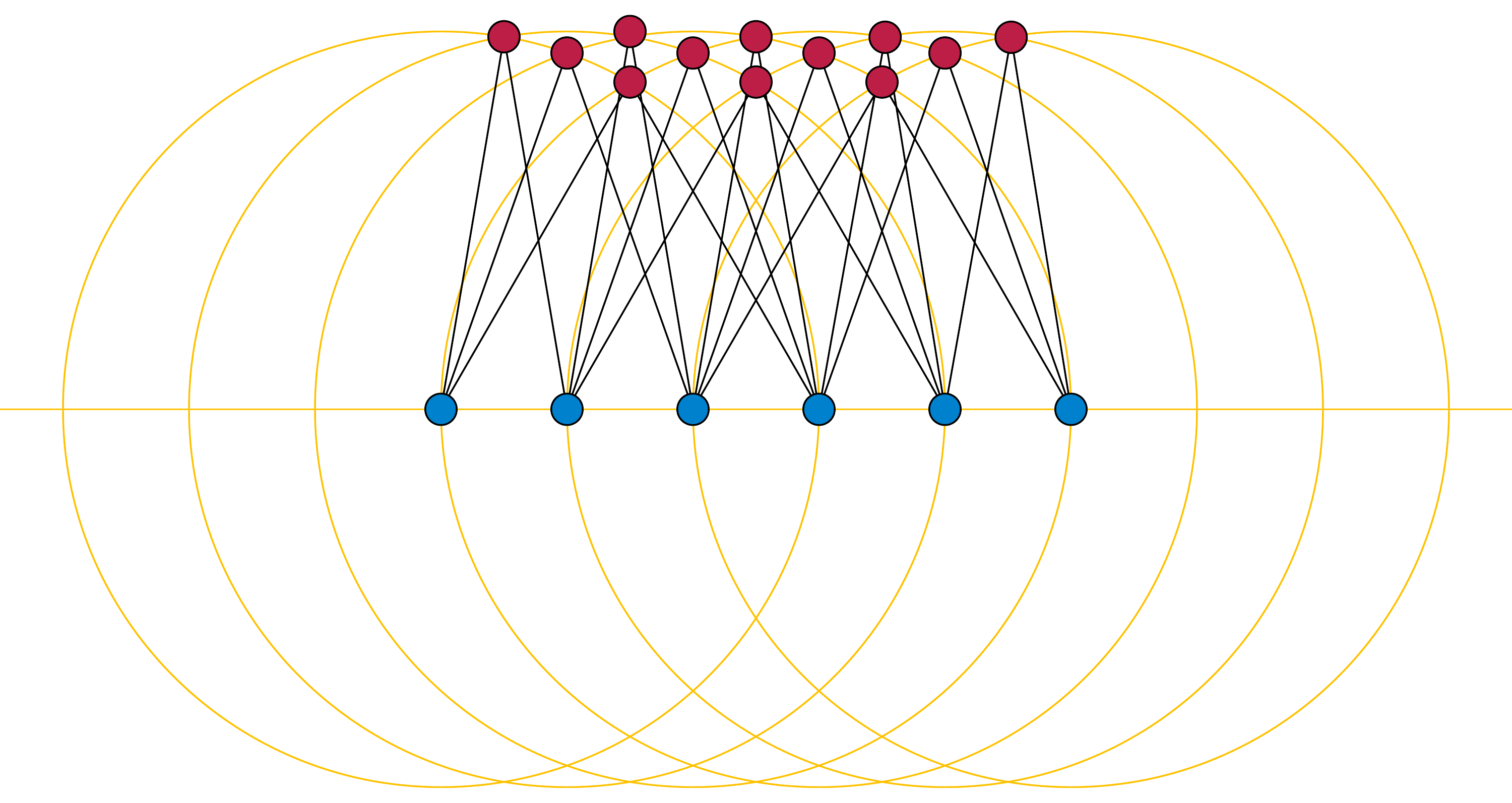}
\caption{Representing an interchange as a unit distance graph or the center points of a unit disk graph.}
\label{fig:unit-distance-interchange}
\end{figure}

\begin{theorem}
The unit distance graphs and unit disk graphs have unbounded flip-width.
\end{theorem}

\begin{proof}
For unit distance graphs, place points representing the lanes equally spaced along a line segment of length less than two in the plane, and place points representing ramps at the intersections of pairs of unit circles centered at the lane points (\cref{fig:unit-distance-interchange}). The resulting graph may have ramp--ramp or lane--lane edges, but it will have no optional lane--ramp edges. For unit disk graphs, scale the same points by a factor of two so that unit disks centered at them will be tangent when their points are adjacent in the unit-distance graph. The resulting unit disk graph includes all possible optional lane--ramp edges, forming an interchange of the same order.
\end{proof}

\begin{theorem}
\label{thm:polyVis}
The visibility graphs of simple polygons do not have bounded flip-width.
\end{theorem}

\begin{proof}
Place points representing lanes on a horizontal line, and place points representing the ramps between two consecutive lanes in the same order on a parallel line above them. Place points representing the remaining ramps, between non-consecutive lanes, on a third parallel line below the lanes. Draw a triangle between each ramp vertex and the two lanes it should connect, and take the union of the triangles. Fill any holes formed in taking the union, keeping only the outer boundary, to form a simple polygon (\cref{fig:visibility-interchange}). In the resulting polygon, each ramp still has parts of two triangle sides adjoining it, blocking its visibility from any lanes that it should not see. Within the triangle for each ramp, it can see its two lanes and any other lane between them.
\end{proof}

\begin{figure}[t]
\centering\includegraphics[width=0.8\columnwidth]{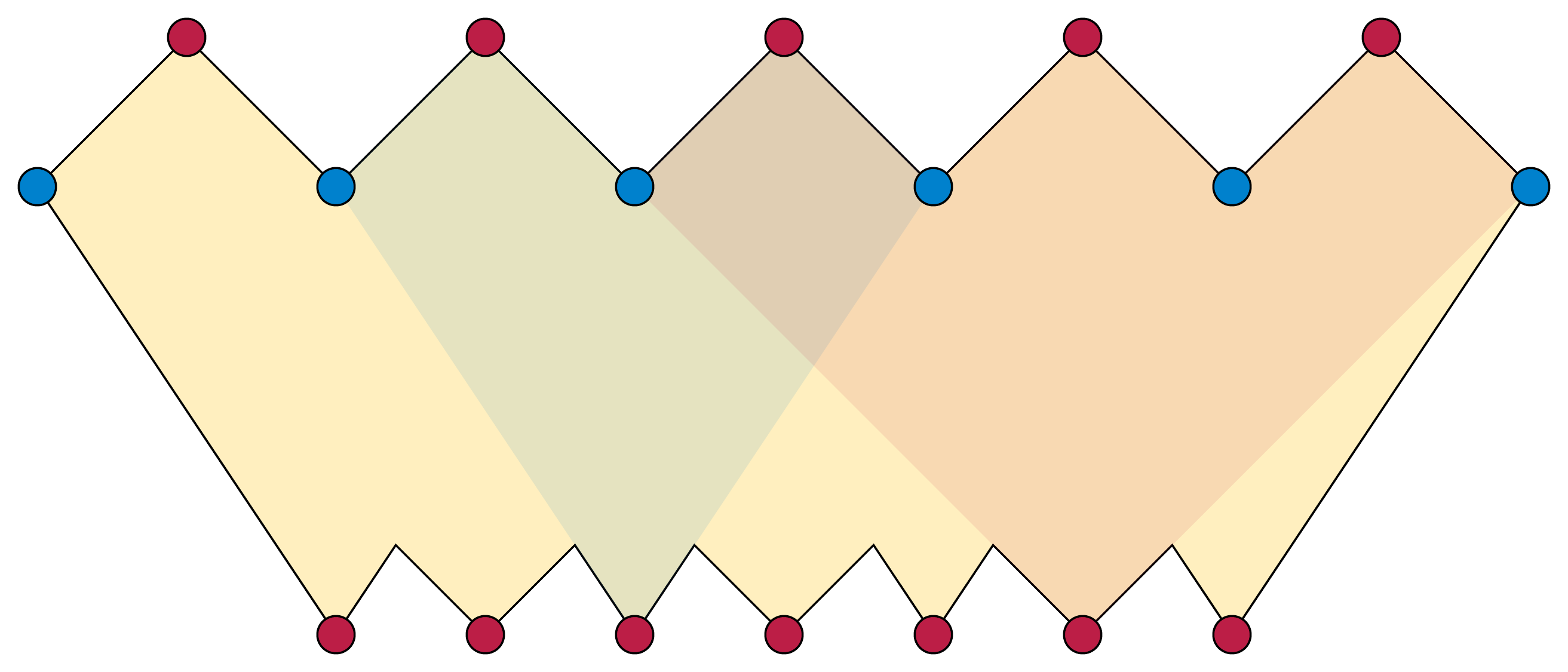}
\caption{Representing an interchange as the visibility graph of a simple polygon.}
\label{fig:visibility-interchange}
\end{figure}

Theorem~\ref{thm:polyVis} also follows from the fact that simple polygon visibility graphs are monadically independent~\cite{BonChaKim-IPEC-22}, which was shown using a similar construction. Furthermore, simple polygon visibility graphs are \emph{cop-win graphs}; a single cop wins a different cop-and-robber game in which both players move along graph edges or stand still~\cite{LubSnoVos-CGTA-17}. But this has no implications for flip-width; adding a universal vertex to any graph makes it cop-win but does not change the boundedness of its flip-width.

The \emph{$\beta$-skeletons} (for $\beta\le 1$) are defined from a set of points by constructing for each pair of points a lune, the intersection of two congruent disks that cross at these points with angle $\pi-\sin^{-1}\beta$. Two points are adjacent if this lune contains no other given points~\cite{KirRad-CG-85}.

\begin{theorem}
For any $\beta<1$, the $\beta$-skeletons have unbounded flip-width.
\end{theorem}

\begin{proof}
Place vertices representing lanes and ramps on the lines $y=0$ and $y=1$ respectively, evenly spaced and very close to the line $x=0$, close enough to ensure that each lane--ramp lune stays within the slab $0\le y\le 1$. For each ramp $r$, place two blocking points on the line $y=1-\varepsilon$ for suitably small $\varepsilon$, close enough to $r$ to avoid all lunes from other ramps. These blocking points should be just outside the two lunes connecting $r$ to its two lanes, one to the left and one to the right, so that any lune connecting $r$ to a lane outside of its range of lanes contains one of the blocking points and is non-empty. The resulting $\beta$-skeleton forms an interchange with all optional lane--ramp edges.
\end{proof}

\begin{observation}
\label{obs:hypercube-interchange}
The graph of a $d$-dimensional hypercube contains an interchange of order $d$.
\end{observation}

\begin{proof}
This is the graph of subsets of a $d$-element set, adjacent when they differ by one element. Let lanes be singletons  and ramps be two-element sets.
\end{proof}

\emph{Rectangle of influence graphs} connect pairs of points in the plane when their bounding box contains no other points~\cite{IchSkl-PR-85,LioLubMei-CGTA-98}. Sources vary on how to treat points on the boundary of this bounding box, but that  can be avoided using point sets with no equal coordinates.

\begin{figure}[t]
\centering\includegraphics[width=0.7\columnwidth]{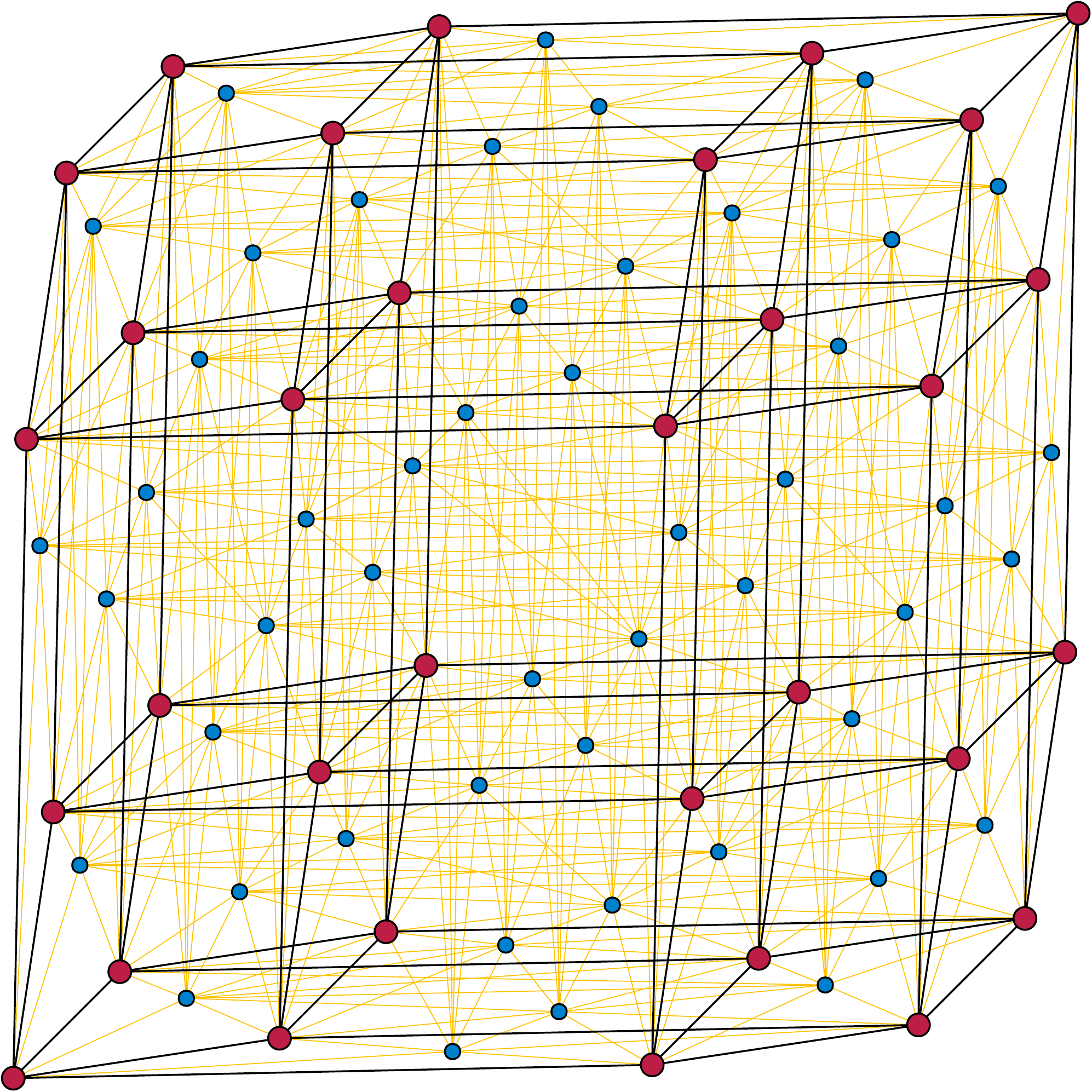}
\caption{A 5-dimensional hypercube graph as an induced subgraph of a rectangle of influence graph.}
\label{fig:empty-rectangle}
\end{figure}

\begin{theorem}
\label{thm:rectangle}
Rectangle of influence graphs induce high-dimension hypercubes and have unbounded flip-width.
\end{theorem}

\begin{proof}
We recursively construct integer points with distinct coordinates whose rectangle of influence graphs contain arbitrarily large induced hypercubes. As a base case, the two points $(0,0)$ and $(1,1)$ give a one-dimensional hypercube graph. If $X_{d-1}$ is defined in this way, with subset $Y_{d-1}$ inducing a $(d-1)$-dimensional hypercube, construct $X_d$ by placing side to side the following three sets: (1) a copy of $X_{d-1}$ scaled vertically by a factor of three; (2) a copy of $Y_{d-1}$ scaled by the same factor, offset vertically by two units (and missing its topmost point), and (3) another scaled copy of $X_{d-1}$, offset vertically by one unit. Choose $Y_d$ to be the copies of $Y_{d-1}$ within the first and third of these three sets. Scaling does not affect the hypercube graphs within these copies. The middle copy of $Y_{d-1}$ blocks all empty rectangles stretching from the first copy to the third copy, except those between corresponding pairs of points, so $Y_d$ induces a $d$-dimensional hypercube graph.
\end{proof}

\cref{fig:empty-rectangle} illustrates five levels of the recursive construction of \cref{thm:rectangle}, with another (cosmetic) step that compacts the coordinates to use consecutive integers.

\begin{theorem}
The graphs of four-dimensional convex polytopes, and of three-dimensional Euclidean Delaunay triangulations, have unbounded flip-width.
\end{theorem}

\begin{proof}
For 4-polytopes, consider the barycentric subdivisions of neighborly polytopes. The graph of a neighborly polytope is complete, and barycentric subdivision preserves realizability as a convex polytope~\cite{EwaShe-MA-74}. The barycentric subdivision replaces the edges of the complete graph by disjoint two-edge paths. The original vertices of the complete graph form the lanes, and the subdivision points of these paths form the ramps, of an interchange, whose order equals the number of vertices in the neighborly polytope.

For Delaunay triangulations, we do not use interchanges; instead we rely on a result of Toruńczyk that weakly sparse graphs (that is, graphs with no $K_{t,t}$ subgraph for some $t$) have bounded flip-width if and only if they have bounded expansion~\cite{Tor-23}. To construct a Delaunay triangulation that does not have bounded flip-width, we begin with the convex hull of certain points in $\mathbb{R}^4$ (coordinatized by pairs of complex numbers), the union of the following three sets of points on a unit sphere, for a given even integer parameter $n$:
\begin{itemize}\setlength\itemsep{0em}
\item The $n$ points $(e^{2\pi i/n},0)$ for integer $i$, $0\le i<n$.
\item The $n$ points $(0,e^{2\pi j/n})$ for integer $j$, $0\le j<n$.
\item The $n^2$ points $(e^{2\pi i/n}/\sqrt2,e^{2\pi j/n}/\sqrt2)$ for $i$ and $j$ in the same range.
\end{itemize}
The points with both coordinates nonzero form a square grid on the flat torus $\{(x,y)\mid |x|=|y|=1/\sqrt2\}$, and the other two subsets of points form $n$-gons (not faces of the convex hull) in the planes $x=0$ and $y=0$. Each edge of an $n$-gon is parallel to the family of edges connecting a ring of squares on the torus, and the facets of the convex hull are warped triangular pyramids  connecting an $n$-gon edge to one of these parallel squares. The subgraph induced in the graph of this polytope by the vertices for which $i$ and $j$ are both even is the subdivision of a complete bipartite graph $K_{n/2,n/2}$. This subdivision has no $K_{2,2}$ subgraph and does not have bounded expansion, so it does not have bounded flip-width.

To transform this inscribed 4-polytope into a Delaunay triangulation, we apply a stereographic projection whose pole is the center point of one of the squares on the torus. On the unit sphere in $\mathbb{R}^4$, this pole belongs to only two of the circumspheres of the facets, for the two facets meeting at this square. Stereographic projection preserves spheres on the unit sphere, so the empty spheres of all of the other facets project to empty spheres for the corresponding set of six points in $\mathbb{R}^3$; that is, these six points form a prism-shaped cell in the Delaunay complex of the projected points. Each edge of the 4-polytope is part of one of these Delaunay cells, so the Delaunay graph of the projected points is the same as that of the complex. Perturbing the points to form a Delaunay triangulation, and keeping only the points for which $i$ and $j$ are both even, again produces the subdivision of a complete bipartite graph, as the subgraph of a three-dimensional Delaunay triangulation.
\end{proof}

\section{Radius-1 flip-width}

Our escape strategy for the robber through interchanges involves the robber taking two steps per move. It is natural to ask whether this can be strengthened to allow escapes of only one edge per move for the same classes of geometric graphs. That is, do these geometric graphs have bounded or unbounded radius-1 flip-width?

In the treewidth game, radius-1 corresponds to \emph{degeneracy}, where a graph has degeneracy $\le d$ if and only if its vertices can be ordered so that each vertex has $\le d$ earlier neighbors. The radius-1 treewidth game can be won by $d+1$ cops, who always play on the current vertex of the robber and its earlier neighbors, forcing the robber to move later in the ordering. On the other hand, if the degeneracy is greater than $d$ then the graph has a $(d+1)$-core, an induced subgraph with minimum degree $d+1$, within which the robber is safe: no matter where the cops move next, there will be an unoccupied vertex for the robber within one step~\cite{Tor-23}.

For the radius-1 flip-width game, Toruńczyk identifies a corresponding concept to a core, which we call a \emph{$\Delta$-diverse subgraph}. This is an induced subgraph in which the open neighborhoods of each two vertices differ by at least $\Delta$ vertices. As Toruńczyk proves, a family of graphs whose graphs contain $\Delta$-diverse subgraphs, for arbitrarily large $\Delta$, has unbounded radius-1 flip-width~\cite{Tor-23}. For cops that make $t$ flips per move, the robber can escape by staying within a $2^{t+1}$-diverse subgraph and moving to a vertex $v$ in this subgraph such that, after the announced flips, $v$ will have at least $2^t$ neighbors. These neighbors, and $v$ itself, form a set of $2^t+1$ vertices within which the robber can move, some two of which (say $x$ and $y$) will be equivalent after the cops' flips. Each vertex adjacent to exactly one of $x$ and $y$ will remain adjacent to exactly one after the flips, so one of $x$ and $y$ will have $\ge 2^t$ neighbors, enough to continue the same strategy.

Slightly more generally, define a $(\Delta,\chi)$-diverse subgraph for a subset of vertices in a given graph and an (improper) $\chi$-coloring of those vertices, keeping all properly colored edges and removing all improperly colored ones, with diversity measured between vertices of the same color. The robber escapes by staying in a $(2^{t+1}(\chi-1)+2,\chi)$-diverse subgraph and moving to a vertex with at least $2^t(\chi-1)+1$ differently-colored neighbors. Some two neighbors will have the same color and be treated equivalently by all flips, and one of these two will have enough neighbors in the next move.

\begin{lemma}
An order-$n$ interchange with all optional lane--ramp edges has an $(\Omega(n^{1/3}),2)$-diverse subgraph.
\end{lemma}

\begin{proof}
Number the lanes of the interchange from $0$ to $n-1$, and choose a parameter $k\approx n^{1/3}$. Form a two-colored subgraph (colored by lanes and ramps) with all lanes and a subset of ramps, the ramps for lanes $x$ and $y$ with $x<y$ that meet the following two conditions: (1) $|y-x|\ge k$, and (2) $x\equiv ky\bmod{k^2+1}$. The second condition leaves $\Theta(n^{1/3})$ ramps starting or ending at each lane, enough so every two lanes have diverse neighborhoods.

\begin{figure}[t]
\centering\includegraphics[width=0.8\columnwidth]{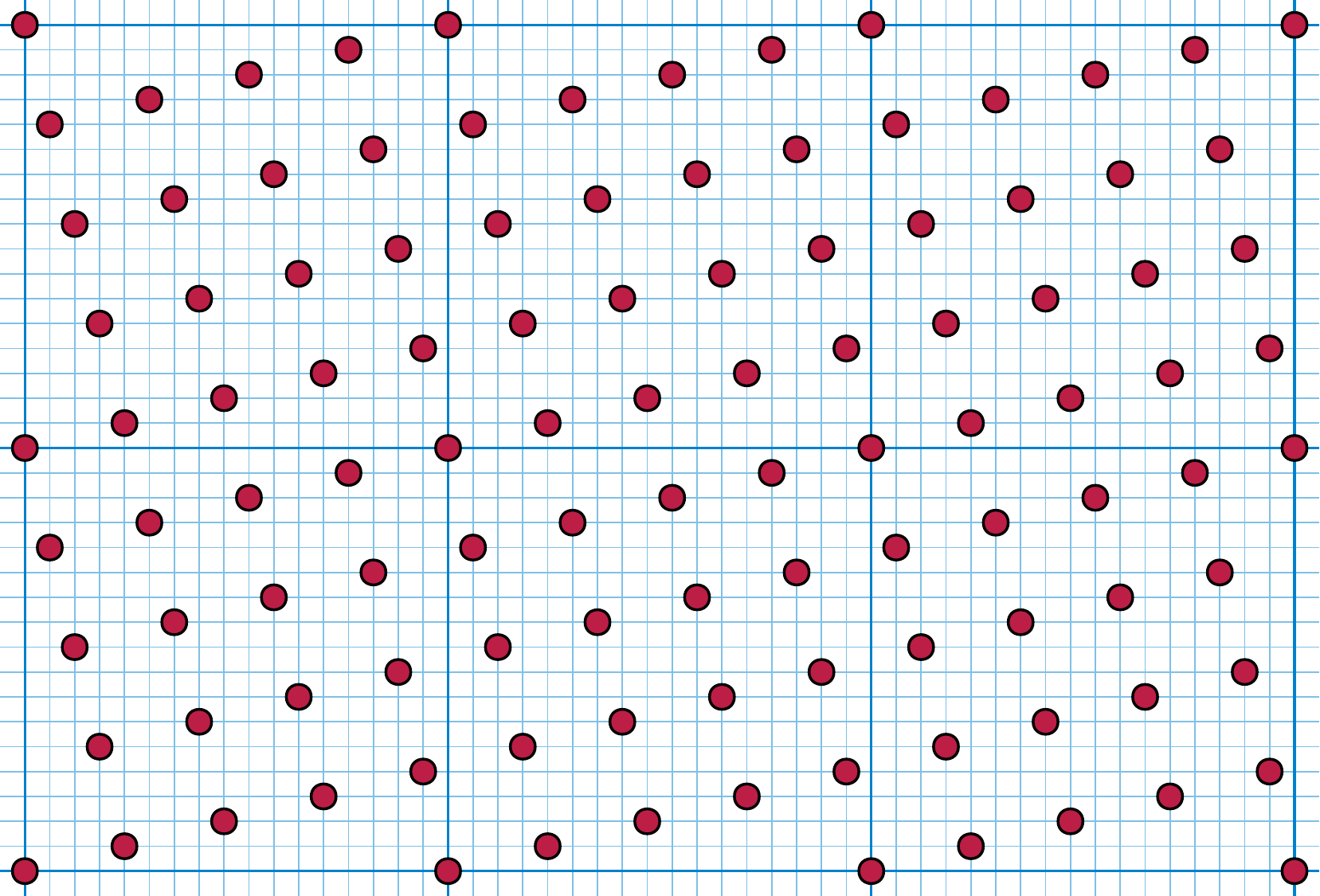}
\caption{The integer points $(x,y)$ with $x\equiv ky\bmod{k^2+1}$ (shown for $k=4$) form a tilted grid in which all pairs of points have $L_1$ distance at least $k+1$.}
\label{fig:tilted-grid}
\end{figure}

Two ramps with disjoint ranges of lanes are distinguished by all their neighbors, of which there are $\Omega(n^{1/3})$ by the first condition. Two ramps with overlapping ranges of lanes $[x,y]$ and $[x',y']$ are distinguished by the lanes between $x$ and $x'$, and the lanes between $y$ and $y'$; there are $|x-x'|+|y-y'|$ such lanes. The condition that $x\equiv ky\bmod{k^2+1}$ defines a subset of the integer grid in the form of a tilted square grid in which all pairs of points are at $L_1$ distance at least $k+1$ (\cref{fig:tilted-grid}), so the number of lanes that are neighbors of only one of the two ramps is also at least $k+1=\Omega(n^{1/3})$.
\end{proof}

\begin{corollary}
Interval graphs, permutation graphs, circle graphs, intersection graphs of axis-aligned unit squares, unit disk graphs, visibility graphs of simple polygons, and $\beta$-skeletons with $\beta<1$ have unbounded radius-1 flip-width.
\end{corollary}

\begin{observation}
The graph of a $d$-dimensional hypercube is $(2d-2)$-diverse.
\end{observation}

\begin{proof}
It is $d$-regular, and the neighborhoods of any two vertices can share at most two neighbors.
\end{proof}

\begin{corollary}
Unit distance graphs, graphs of 4-polytopes, and rectangle of influence graphs have unbounded radius-1 flip-width.
\end{corollary}

\begin{proof}
The $d$-dimensional hypercube graphs can be drawn as unit distance graphs by linear projections that map the basis vectors of $\mathbb{R}^d$ to generic unit vectors in the plane~\cite{Bra-CGTA-96}. They are the graphs of the 4-dimensional \emph{neighborly cubical polytopes}~\cite{JosZie-DCG-00}. For rectangle of influence graphs, the result follows from \cref{thm:rectangle}
\end{proof}

We do not know whether three-dimensional Delaunay triangulations have bounded radius-1 flip-width; we leave this as open for future research.

\bibliographystyle{plainurl}
\bibliography{flipwidth}

\newpage\appendix
\section{Monadic dependence}
\label{sec:monadicNIP}

In this appendix we show that graph classes that include arbitrarily large interchanges are monadically independent. This in particular implies that they do not have bounded flip-width, but unlike \cref{thm:radius-2} it does not directly bound the radius $r$ at which the robber can win.

Here, \emph{monadic dependence} is a general notion for logical relations of arbitrary arity based on the concept of a \emph{transduction}, a method of representing one logical structure by a system of first-order formulas over a finite system of monadic predicates, or equivalently over a finite coloring, of a second structure~\cite{BraLas-TAMS-21}. When applied to the first-order theory of graphs, a transduction is a mapping from vertex-colored graphs to graphs, defined by logical predicates that describe which pairs of vertices are adjacent in the image graph, and which vertices of the starting graph correspond to vertices in the image graph~\cite{NesOssSie-CSL-22}. A family of graphs $\mathcal{F}$ is monadically dependent (also written as ``monadically NIP'') if every transduction on the graphs in $\mathcal{F}$ is incomplete: there is some target graph that it does not produce, regardless of which graph in~$\mathcal{F}$ and which coloring of that graph it is applied to. On the other hand, $\mathcal{F}$ is monadically independent (``not monadically NIP'') if it is possible to find a ``universal'' transduction, one that maps colorings of graphs in $\mathcal{F}$ to all possible graphs.

It is conjectured that, among hereditary graph families, the monadically dependent families are exactly the ones for which first-order model checking is fixed-parameter tractable~\cite{ALS-16,GajHliObd-TCL-20}. For hereditary classes of ordered graphs, and logical properties that can make use of this ordering, being monadically dependent is equivalent to having bounded twin-width~\cite{BonGiodeM-STOC-22}. Without the assumption of an ordering, bounded flip-width implies monadically dependence, and a weakening of bounded flip-width, \emph{almost bounded flip-width}, is conjectured to be equivalent to monadic dependence~\cite{Tor-23}.


In this section we use a slightly stronger definition of interchanges that includes ramps between all pairs of lanes, instead of omitting the outermost lanes. An interchange under this stronger definition satisfies the weaker definition (where fewer ramps are required). In the other direction, an interchange for the weaker definition can be converted to one for the stronger definition, with two fewer lanes, by omitting the two outermost lanes.

We first observe that every huge interchange contains, as an induced subgraph,
a large interchange in which either all optional ramp-lane edges are present,
or a large interchange in which no optional ramp-lane edges are present. Call those two types of interchanges \emph{dense interchanges} and \emph{sparse interchanges}, respectively.

\begin{lemma}
\label{lem:pure}
If a class of graphs contains arbitrarily large interchanges, it contains either arbitrarily large dense interchanges or arbitrarily large sparse interchanges.
\end{lemma}

\begin{proof}
This follows from applying standard methods of Ramsey theory to a 2-colored 3-uniform hypergraph on triples of lanes, with each triple given color $1$ if an optional edge connects the middle of the three lanes to the ramp between the two outer lanes, and color $0$ if the optional edge is not present. By Ramsey's theorem, we can extract a large subset of the lanes that induces a monochromatic sub-hypergraph. The selected lanes, together with the ramps associated with pairs of selected lanes, then form a large regular interchange: if they all have color $1$, the result is a dense interchange, and if they all have color $0$, the result is a sparse interchange.
\end{proof}

The following result is from Szymon Toruńczyk (personal communication):

\begin{theorem}
A graph class that contains arbitrarily large interchanges as induced subgraphs is monadically independent.
\end{theorem}

\begin{proof}
By \cref{lem:pure}, we can assume either that there are arbitrarily large sparse interchanges, or that there are arbitrarily large dense interchanges. We show that in each of these cases, the class is monadically independent.

We first argue that any class containing arbitrarily large sparse interchanges is monadically independent.
We can represent an arbitrary $n$-vertex graph $G$ in a colored sparse interchange with at least $n$ lanes,
by identifying the vertices of $G$ with a subset of the lanes, and for each ramp associated 
with a pair $uv$ of vertices of $G$, marking this ramp with a special color if and only if $u$ and $v$ are adjacent in $G$. Formally, describe this marking by a monadic predicate $A(z)$ that is true of these marked ramps and false for all other vertices. Then adjacency in the original graph $G$ can be recovered from the obtained colored interchange,
by a first-order formula
\[\phi(x,y)\equiv \exists z\bigl(
x\sim z \wedge y\sim z \wedge A(z) \bigr)
\]
(where adjacency in $G$ is represented by the binary predicate $\sim$), expressing that there is a vertex $z$ that is adjacent to $x$ and $y$ and has the special marking. 
The formula $\phi$ does not depend on the selected graph $G$.
It follows that the formula $\phi$ can be used to define a transduction that
may produce an arbitrary graph $G$ from any sufficiently large sparse interchange,
by first coloring its vertices using a single special color,
then applying the formula $\phi(x,y)$ to define a new edge relation,
and finally, taking an induced subgraph to restrict only to the lanes of the interchange that correspond to the vertices of $G$.
Thus, every class transduces the class of all graphs, and is therefore monadically independent.

We now argue that every class containing arbitrarily large dense interchanges is monadically independent,
by showing that such a class transduces a class of arbitrarily large sparse interchanges.
Since transductions can be composed, this implies that the former class is monadically independent.

Given a dense interchange, color all the ramps using a special color, denoted by the monadic predicate $R(v)$,
and color the first lane (in the order on lanes) with another color, denoted by the monadic predicate $F(v)$.
We can recover the order on two lanes $x$ and $y$ using a first-order formula
\[
\begin{split}
\psi(x,y)&\equiv\\
&\bigl(F(x)\wedge x\ne y\bigr) \vee{}\\
&\exists w \exists z \bigl(
F(w) \wedge R(z) \wedge w\sim z \wedge x\sim z \wedge \lnot(y\sim z)
\bigr).\\
\end{split}
\]
Namely, a lane $x$ is smaller than a lane $y$ 
if and only if either $x$ is the first lane, or there a ramp $z$ which is adjacent to the first lane and 
to $x$ but not to $y$. Since ramps and the first lane are marked 
with a special color, this can be expressed using a fixed first-order formula
(not depending on the interchange), that can use the colors.

Now, we can identify the two lanes $x,y$ to which a ramp $z$ is associated with,
using a first-order formula.
Namely, in a sparse interchange, a ramp $z$ should be associated with exactly the smallest lane,
and the largest lane among its neighboring lanes. 
Therefore, $z$ is associated with a lane $x$ if and only if 
$z$ is adjacent to $x$ and there do not exist lanes $u$ and~$v$, smaller and larger than $x$, respectively, that are also adjacent to $z$.
This can be expressed by a first-order formula
\[
\begin{split}
\gamma(x,z)\equiv{}& \lnot R(x) \wedge R(z) \wedge x\sim z\wedge{}\\
&\lnot\exists u\exists v
\left(
\begin{split}
&\lnot R(u) \wedge \lnot R(v) \wedge{}\\
&u\sim z\wedge v\sim z \wedge{}\\
&\psi(u,x) \wedge \psi(x,v) \\
\end{split} 
\right).
\end{split}
\]
that uses the colors, and does not depend on the considered interchange.
Thus, the formula $\gamma(x,z)$ defines 
the edges of a sparse interchange, of the same size as the original dense interchange.

Hence, every class that contains arbitrarily large dense interchanges as induced subgraphs
transduces a class that contains arbitrarily large sparse interchanges, and is therefore monadically independent.
\end{proof}

This applies, in particular, to all the classes of geometric graphs for which we have constructed large interchanges.
Our proof that three-dimensional Delaunay triangulations have unbounded flip-width is an exception, as it does not use interchanges. However, in this case we have constructed weakly sparse Delaunay triangulations that are not nowhere-dense, so it follows from known results that they are monadically independent.

\end{document}